\documentclass[12pt]{article}

\usepackage[T1]{fontenc}
\usepackage[utf8]{inputenc}
\usepackage[letterpaper,margin=1in]{geometry}
\usepackage{setspace}
\usepackage{amsmath,amssymb,amsthm,mathtools,bm}
\usepackage{booktabs,array,placeins,flafter}
\usepackage[authoryear]{natbib}
\usepackage[colorlinks=true,linkcolor=blue,citecolor=blue,urlcolor=blue]{hyperref}
\hypersetup{ pdftitle={Representation Multiplicity in Causal Forests}, pdfauthor={Yi Niu}, pdfkeywords={causal forests, feature subsampling, randomized search, representation invariance, redundant covariates} }

\allowdisplaybreaks
\newtheorem{assumption}{Assumption}
\newtheorem{theorem}{Theorem}
\newtheorem{proposition}{Proposition}
\newtheorem{corollary}{Corollary}
\newcommand{\E}{\mathbb E}
\newcommand{\Prb}{\mathbb P}
\newcommand{\R}{\mathbb R}
\newcommand{\one}{\mathbf 1}
\newcommand{\Pois}{\operatorname{Poisson}}
\newsavebox{\mcgridbody}
\newsavebox{\mcgridpage}
\newsavebox{\empiricalbody}
\newsavebox{\empiricalpage}
\newcommand{\legend}[1]{\par\medskip\begin{minipage}{\linewidth}\footnotesize \textit{Notes.} #1 \end{minipage}}

\begin{document}
\title{Representation Multiplicity in Causal Forests\thanks{I thank Karun Adusumilli, Xu Cheng, Wayne Yuan Gao, and Petra Todd for helpful comments and suggestions.}}
\author{Yi Niu\thanks{Department of Economics, University of Pennsylvania. Email \texttt{n1y1@sas.upenn.edu}.}}
\date{September 2026}
\maketitle

\begin{abstract}
\noindent Including covariates alongside strictly monotone transformations can change a causal forest's treatment decisions without adding information. Random feature selection favors covariates represented by these multiple columns. Under stated conditions, I show that this imbalance can persist as samples grow. Treatment effect components associated with other covariates are omitted, attenuated, or recovered depending on their inclusion probabilities and tree depth. Simulations and a job-training replication illustrate sensitivity to redundant transformations. Sampling groups of variables that generate identical splits restores prediction invariance on the grouping data when fitting and randomization are held fixed.
\medskip


\end{abstract}

\newpage

\section{Introduction}\label{sec:intro}

Suppose two researchers use a causal forest to estimate how a job-training program affects later earnings as a function of baseline household income and age. One includes income and age. The other keeps age but includes income, log income, and several other strictly increasing transformations of income. If baseline income has positive support, these income columns generate the same threshold partitions and contain the same covariate information. A feature-subsampled forest nevertheless treats them as separate candidates, making an income split more likely to enter the node-level search and an age split less likely.

Age can shape a partition only when it is sampled and used before the path terminates. Redundant income columns reduce age's inclusion probability, while deeper paths provide more opportunities to use it.

The main result studies honest forests of treatment-effect scores with two independent covariates and binary conditional average treatment effect (CATE) components of different magnitudes. The paired representations assign multiple strictly monotone transformations to one covariate and a single column to the other, preserving the sigma-field and split-action library. Let \(H_n\) denote effective path depth, \(M_n\) the number of columns, and \(Q_{M_n}\) the number sampled at a node. Their interaction is summarized by the access-depth index \(\mathcal I_n=H_n\E[Q_{M_n}]/M_n\).

Under the regularity and continuation conditions in Section~\ref{sec:model}, \(\mathcal I_n\to0\) makes each forest omit its singleton component. The paired CATE limits then have opposite signs on half the population, with nonvanishing policy regret when the smaller component receives multiple columns. If \(\mathcal I_n\to\infty\), both components are recovered. An intermediate limit yields attenuation under the additional path-depth condition in Theorem~\ref{thm:main}. Recovery thus holds for each fixed finite representation without holding uniformly over growing equivalent representations.

In the \texttt{grf} R package, \texttt{mtry} controls the number of sampled split candidates. With \texttt{mtry} fixed, the index has order \(H_n/M_n\). Under the default square-root scaling, it has order \(H_n/\sqrt{M_n}\). The corresponding representation-growth boundaries are \(H_n\) and \(H_n^2\).

The remedy is to sample split-action classes. On a declared finite array, columns with the same weak ordering and ties generate the same threshold actions for classification and regression trees (CART). Sampling these classes gives identical predictions across equivalent representations on that array under canonical routing and common class-level randomness.

A six-covariate Monte Carlo varies multiplicity, \texttt{mtry}, and minimum node size while holding covariate information fixed. With tenfold representation and minimum node size five, paired treatment rules disagree on 32.9 percent of evaluation observations with \(\texttt{mtry}=2\), compared with 1.2 percent under the package default. Larger minimum node sizes shorten paths and increase disagreement when candidate access is limited. Class sampling eliminates representation disagreement and reproduces the performance of the forest without repeated columns. A replication of \citet{chernozhukov2022automatic} reduces raw-square pairs to single columns in the random-forest outcome regression. It reveals changes in individual regression-contrast signs despite stable aggregate treatment-effect conclusions.

The paper builds on honest causal forests \citep{wager2018} and generalized random forests \citep{athey2019}. Honest sample separation controls the use of outcomes in leaf estimation, while feature subsampling determines which covariates can shape those leaves.

\citet{louppe2014} explains how redundant predictors change their probability of entering the split search and can affect prediction. \citet{chi2022high} and \citet{klusowski2024large} provide access-dependent recovery bounds, while \citet{mei2026exogenous} characterize bias from unresolved components under training-data-independent partitions. I connect representation-induced changes in split access to limiting CATE estimates and treatment-rule regret in sample-grown honest forests.

\citet{chen2013module} sample modules before representative variables, while \citet{colot2021} weight features inversely by cluster size. Here, columns are grouped by equality of their unordered threshold partitions on a declared array. Proposition~\ref{prop:class} establishes exact invariance on this array under canonical routing and common class-level randomness.

Section~\ref{sec:model} introduces the model, and Section~\ref{sec:results} derives the forest limits and policy consequences. Section~\ref{sec:quotient} presents class sampling, Section~\ref{sec:evidence} reports the evidence, and Appendix~\ref{app:proofs} contains the proofs.

\section{Model Setup}\label{sec:model}

Let \(X=(X_1,X_2)\sim\operatorname{Unif}([0,1]^2)\) be the covariate vector. For \(x=(x_1,x_2)\in[0,1]^2\) and \(j\in\{1,2\}\), define the binary step functions \(\psi_j(x)=2\one\{x_j\ge1/2\}-1\), and define the CATE
\[
\tau_a(X)=a_1\psi_1(X)+a_2\psi_2(X),
\]
indexed by the coefficient vector \(a=(a_1,a_2)\). Restrict \(a\) to the parameter class
\[
\mathcal C_a=
\{(a_1,a_2)\mid\underline a\le a_1<a_2\le\overline a,
\ a_2-a_1\ge\Delta_a\}
\]
where \(\underline a\) and \(\overline a\) are fixed lower and upper coefficient bounds satisfying \(0<\underline a<\overline a\), and \(\Delta_a\) is a fixed minimum coefficient gap, with \(0<\Delta_a<\overline a-\underline a\). These restrictions keep the sign and regret comparisons nondegenerate uniformly over \(\mathcal C_a\).

Let \(n\) be the sample size, and let \((Y_i,W_i,X_i)_{i=1}^n\) denote the i.i.d. observations. For treatment state \(w\in\{0,1\}\), the potential outcome \(Y(w)\) follows
\begin{equation}\label{eq:dgp}
Y(w)=m_0(X)+w\tau_a(X)+U,
\qquad
\E(U\mid X,W)=0,
\end{equation}
where \(m_0\) is the baseline outcome function and \(U\) is the outcome disturbance. The binary treatment indicator \(W\) is randomized with propensity score one half, and the observed outcome \(Y\) satisfies
\[
W\perp\{Y(0),Y(1)\}\mid X,
\qquad
\Prb(W=1\mid X)=\frac12,
\qquad
Y=Y(W).
\]
The baseline function \(m_0\) and the conditional law of \(U\) are held fixed over \(a\in\mathcal C_a\). Define the observed-outcome regression \(m_a(X)=\E(Y\mid X)=m_0(X)+\tau_a(X)/2\). The oracle transformed-outcome score for observation \(i\) is
\begin{equation}\label{eq:score}
\Gamma_i=4(W_i-1/2)\{Y_i-m_a(X_i)\}
\end{equation}
and satisfies \(\E(\Gamma_i\mid X_i=x)=\tau_a(x)\) at every covariate value \(x\). Let \(\widehat m_n\) be an estimator of \(m_a\), and define
\[
\widehat\Gamma_{i,n}
=4(W_i-1/2)\{Y_i-\widehat m_n(X_i)\},
\qquad
\epsilon_n=\|\widehat m_n-m_a\|_\infty.
\]

For covariate index \(j\in\{1,2\}\), let \(c_n\in\{1,2,\ldots\}\) be its representation multiplicity. Choose \(c_n\) strictly monotone transformations \(\{\varphi_{j\ell n}\mid\ell=1,\ldots,c_n\}\), including the identity, and denote the \(\ell\)-th transformed column by \(X_{j,n}^{(\ell)}=\varphi_{j\ell n}(X_j)\). Define the paired covariate representations
\begin{align*}
  R_{1,n}(X)=(X_{1,n}^{(1)},\ldots,X_{1,n}^{(c_n)},X_2),\qquad R_{2,n}(X) =(X_1,X_{2,n}^{(1)},\ldots,X_{2,n}^{(c_n)}).
\end{align*}
Each representation has \(M_n=c_n+1\) raw columns. Both representations generate the covariate sigma-field \(\sigma(X_1,X_2)\). If \(\varphi_{j\ell n}\) is increasing, a split \(X_j\le t\) corresponds to \(\varphi_{j\ell n}(X_j)\le\varphi_{j\ell n}(t)\). If it is decreasing, the inequality is reversed. Both cases generate the same unordered threshold partition in exact arithmetic\footnote{Under the same forest conditions, the arguments extend to split-equivalent transformations that generate the same unordered threshold partitions on the covariate support. I use strictly monotone transformations and the identity mapping for readability.}.

Let \(s_n\) be the tree subsample size. Each tree draws \(s_n\) observations uniformly without replacement and randomly divides them into splitting and estimation samples, independently of their values. Both sample fractions are at least a fixed \(\eta\in(0,1/2)\). Write \(n_{C,n}\) for the splitting-sample count in node \(C\). At each node, the tree draws a candidate count \(Q_{M_n}\in\{1,\ldots,M_n\}\) and samples that many raw columns uniformly without replacement, independently of the data and previous candidate draws. It evaluates every distinct threshold partition of the sampled columns. With minimum child splitting-sample size \(k_n\) and split-balance parameter \(\beta\in(0,1/2)\), a split is admissible when each child contains at least \(\max\{k_n,\beta n_{C,n}\}\) splitting observations. The tree maximizes the score gain in \eqref{eq:pseudo-criterion} over admissible splits and stops when no such split has positive gain. No additional depth cap or stopping rule is imposed. Each leaf prediction is the mean fitted score of its estimation observations, with value zero assigned to an empty leaf.

Conditional on the training data, fitted nuisance functions, and root sample assignments, each node's action depends only on its ancestral states, its observations, and fresh randomness assigned to that node. Randomization is independent across nodes.

Write \(P_X\) for the covariate distribution and \(\|f\|_{L^1(P_X)}=\int |f(x)|\,dP_X(x)\). For a split of cell \(C\) into \(C_L\) and \(C_R\), indexed by column \(j\) and threshold \(t\), define the score gain
\begin{equation}\label{eq:pseudo-criterion}
\widehat Q_{C,n}(j,t)
=\widehat p_{L\mid C,n}\widehat p_{R\mid C,n}
\left(\overline{\widehat\Gamma}_{C_L,n}
-\overline{\widehat\Gamma}_{C_R,n}\right)^2.
\end{equation}
The shares and score means are computed from the splitting observations in the corresponding cells.

The total reduction in squared-error loss from a split is \(n_{C,n}\widehat Q_{C,n}(j,t)\). Under representation \(R_{h,n}\), let \(\widehat T_{h,n}\) be the average prediction from \(B_n\) trees constructed by this rule.

To study paths that have not used the singleton, define a locally conditioned failure tree. It draws its root samples by the ordinary rule and conditions each node's transition on making no split on coordinate \(3-h\), retaining stopping as a possible outcome. This construction is used whenever the conditioning probabilities are positive.

\begin{assumption}\label{ass:regularity}
All probability statements and rates hold uniformly over \(a\in\mathcal C_a\).
\begin{enumerate}
\item[(i)] \textnormal{Regularity and rates.}
Outcomes are uniformly bounded, the conditional law of \(U\) has a continuous density on bounded support, and \(\epsilon_n=o_p(1)\). For fixed \(\bar s<1\),
\begin{equation}\label{eq:rates}
s_n/n\le\bar s,\qquad
k_n\to\infty,\qquad k_n/s_n\to0,\qquad
\frac{\log(B_ns_n)}{k_n}\to0,\qquad B_n\to\infty.
\end{equation}

\item[(ii)] \textnormal{Tree randomization.}
Conditional on the training data and fitted nuisance functions, trees have identical marginal distributions and form independent groups of uniformly bounded size.

\item[(iii)] \textnormal{Continued splitting.}
Fix \(C_0>4\). With probability tending to one, every raw column admits a positive-gain admissible split at every reached node with \(n_{C,n}\ge C_0k_n\). This holds for ordinary trees and for locally conditioned failure trees whenever defined.
\end{enumerate}
\end{assumption}

Suppose a tree has already split on \(X_1\), while \(X_2\) remains unused. If the next candidate draw offers only \(X_1\) and its transformations, Assumption~\ref{ass:regularity}(iii) guarantees another admissible positive-gain split whenever \(n_{C,n}\ge C_0k_n\). Further splits provide new candidate draws, giving \(X_2\) more chances to be offered. Together with split balance, this condition makes the number of draws along a path of order \(\log(s_n/k_n)\).

Define the effective path depth \(H_n\), the singleton's per-node inclusion probability \(\omega_n\), and the access-depth index \(\mathcal I_n\) by
\begin{equation}\label{eq:access-index}
H_n=\log(s_n/k_n),
\qquad
\omega_n=\frac{\E Q_{M_n}}{M_n},
\qquad
\mathcal I_n=H_n\omega_n.
\end{equation}

\section{Forest Limits and Treatment-Rule Distortion}\label{sec:results}

The singleton's accumulated split access determines whether its CATE component survives in the forest limit. At intermediate access, its weight also depends on the number of searches a path permits before first use.

For \(\omega_n<1\), the failure-tree transitions are well-defined because the probability of making no singleton split is at least \(1-\omega_n>0\). Let \(J_{h,n}^{\mathrm{fail}}(X)\) count candidate draws along its path to \(X\) at nodes with \(n_{C,n}\ge2k_n\), including a final unsuccessful search if one occurs. For the intermediate limit, assume for fixed \(a\in\mathcal C_a\) that
\begin{equation}\label{eq:intermediate-depth}
\frac{J_{h,n}^{\mathrm{fail}}(X)}{H_n}\longrightarrow_p d_h(a)\in(0,\infty),
\qquad h=1,2.
\end{equation}
Convergence is joint over the training sample, the failure tree, and an independent \(X\sim P_X\).

\begin{theorem}[Representation-dependent forest limits]\label{thm:main}
Under the model and forest construction in Section~\ref{sec:model} and Assumption~\ref{ass:regularity}, the two endpoint results hold uniformly over \(a\in\mathcal C_a\). If \(\mathcal I_n\to0\), then
\begin{equation}\label{eq:selection-limit}
\max_{h=1,2}\|\widehat T_{h,n}-a_h\psi_h\|_{L^1(P_X)}\longrightarrow_p0.
\end{equation}
If \(\mathcal I_n\to\infty\), then
\begin{equation}\label{eq:recovery-limit}
\max_{h=1,2}\|\widehat T_{h,n}-\tau_a\|_{L^1(P_X)}\longrightarrow_p0.
\end{equation}
For fixed \(a\in\mathcal C_a\), if \(\mathcal I_n\to\kappa\in(0,\infty)\) and \eqref{eq:intermediate-depth} holds, then
\begin{equation}\label{eq:intermediate-targets}
\begin{aligned}
\widehat T_{1,n}&\longrightarrow_p
 a_1\psi_1+\{1-e^{-\kappa d_1(a)}\}a_2\psi_2,\\
\widehat T_{2,n}&\longrightarrow_p
 \{1-e^{-\kappa d_2(a)}\}a_1\psi_1+a_2\psi_2
\end{aligned}
\qquad\text{in }L^1(P_X).
\end{equation}
\end{theorem}

The factor \(e^{-\kappa d_h(a)}\) is the limiting probability of reaching a leaf without using the singleton. Its complement is the weight on that component. Both endpoint limits and this attenuation apply to exact copies and strictly monotone transformations under the stated conditions.

For a binary treatment rule \(\pi\), define policy value and regret by
\[
V_a(\pi)=\E\{\tau_a(X)\pi(X)\},
\qquad
\operatorname{Reg}_a(\pi)=V_a(\pi_a^*)-V_a(\pi).
\]
The optimal unconstrained rule and the fitted rules are
\[
\pi_a^*(x)=\one\{\tau_a(x)>0\}=\one\{\psi_2(x)=1\},
\qquad
\widehat\pi_{h,n}(x)=\one\{\widehat T_{h,n}(x)>0\}.
\]
Write \(\Prb_X\) for probability over an independent covariate draw, conditional on the fitted forests.

\begin{corollary}[Sign reversal and policy regret]\label{cor:sign}
Under the conditions of Theorem~\ref{thm:main}, if \(\mathcal I_n\to0\), then
\begin{equation}\label{eq:sign-reversal}
\Prb_X\!\left\{\widehat T_{1,n}(X)\widehat T_{2,n}(X)<0,
\min_{h=1,2}|\widehat T_{h,n}(X)|\ge\underline a/2\right\}
\longrightarrow_p\frac12.
\end{equation}
Moreover,
\begin{equation}\label{eq:regret}
\operatorname{Reg}_a(\widehat\pi_{2,n})\longrightarrow_p0,
\qquad
\operatorname{Reg}_a(\widehat\pi_{1,n})
\longrightarrow_p\frac{a_2-a_1}{2}.
\end{equation}
\end{corollary}

Thus representing the covariate with the larger coefficient by multiple transformations yields the optimal unconstrained rule in the selection regime. Doing so for the covariate with the smaller coefficient yields a persistent treatment mistake.

In practice, the \texttt{grf} 2.4.0 algorithm reference \citep{grfreference} specifies the candidate-count law as
\begin{equation}\label{eq:poisson-count}
N_M\sim\Pois(\xi_M),
\qquad
Q_M=\max\{1,\min(N_M,M)\},
\qquad
\xi_M=\texttt{mtry}
\end{equation}
Here \(N_M\) is the uncapped Poisson candidate count, \(\xi_M\) is its mean parameter, and \(Q_M\) is the candidate count after truncation to \(\{1,\ldots,M\}\).

\begin{corollary}[Representation-growth rates]\label{cor:boundary}
Under the conditions of Theorem~\ref{thm:main}, if \(M_n\to\infty\) and \(\xi_{M_n}\equiv\xi>0\), then
\[
\omega_n\sim\frac{\xi+e^{-\xi}}{M_n},
\]
so the critical representation multiplicity is \(c_n\asymp H_n\). Under the package default \(\xi_M=\min\{\lceil\sqrt M\rceil+20,M\}\), \(\omega_n\sim M_n^{-1/2}\), and the critical representation multiplicity is \(c_n\asymp H_n^2\). Every fixed finite multiplicity satisfies \(\mathcal I_n\to\infty\).
\end{corollary}

Under the theorem's conditions, the default allows \(c_n=o(H_n^2)\), compared with \(c_n=o(H_n)\) for fixed \texttt{mtry}. Both bounds grow slowly because \(H_n=\log(s_n/k_n)\).

\section{Finite-Array Class Sampling}\label{sec:quotient}

Sampling columns can favor a covariate represented by several split-equivalent transformations. Class sampling removes this dependence on representation multiplicity by grouping columns with identical unordered threshold partitions and sampling these groups at each node. Let \(\mathcal X_N=\{x_1,\ldots,x_N\}\) be a labeled array containing every point used for fitting, prediction, variance estimation, and policy ranking. Let \(R\in\R^{N\times p}\) be a representation of this array. For column \(j\), let \(r_j\) be its weak-rank vector, using consecutive ranks \(1,\ldots,m_j\), where \(m_j\) is its number of distinct values. Define
\[
c_j=\min_{\rm lex}\{r_j,(m_j+1)\mathbf 1_N-r_j\},
\qquad
\mathcal C(R)=\{c_j\mid1\le j\le p\}.
\]
Here \(\mathbf 1_N\) is the vector of ones. The canonical matrix \(R^{\rm can}\) contains the distinct vectors in \(\mathcal C(R)\) as columns, ordered lexicographically. Columns with equal \(c_j\) have the same ties and generate the same unordered threshold partitions on every subset of the array.

At each node, sample classes uniformly without replacement and evaluate CART splits on \(R^{\rm can}\). Use common class-indexed randomness for candidate draws, child orientation, split ties, sample allocations, and all remaining forest steps. Hold nuisance fits fixed or fit them on the canonical matrix with common randomness. Apply tuning rules to that matrix, and break policy-ranking ties by a fixed observation order.

\begin{proposition}[Class-sampled forest]\label{prop:class}
Let \(R\in\R^{N\times p}\) and \(R'\in\R^{N\times p'}\) satisfy \(\mathcal C(R)=\mathcal C(R')\). Apply the procedure above with common outcomes, treatment indicators, settings, and randomization \(\theta\). Write \(\mathcal P_b(R,\theta)\) for tree \(b\)'s ordered node partitions, \(\boldsymbol\alpha(R,\theta)\in\R^{N\times n}\) for the forest-weight matrix, and \(\mathcal F(R,\theta)\) for the predictions, reported standard errors, and selected sets on \(\mathcal X_N\). For every \(\theta\) for which the procedure is defined,
\[
\begin{aligned}
\mathcal P_b(R,\theta)&=\mathcal P_b(R',\theta),\qquad b=1,\ldots,B,\\
\boldsymbol\alpha(R,\theta)&=\boldsymbol\alpha(R',\theta),\qquad
\mathcal F(R,\theta)=\mathcal F(R',\theta).
\end{aligned}
\]
\end{proposition}

Proposition~\ref{prop:class} establishes exact representation invariance for the class-sampled forest. The result covers column permutations, additions or deletions that retain every class, and transformations preserving or reversing weak order and ties. For predictions at additional points, class equivalence must hold on the enlarged array.

\section{Finite-Sample Evidence}\label{sec:evidence}

\subsection{Monte Carlo}

The Monte Carlo examines how \texttt{mtry} and minimum node size affect sensitivity to representation multiplicity. The six underlying covariates and the treatment-effect function remain fixed. Draw \(X\sim\operatorname{Unif}([0,1]^6)\), \(W\sim\operatorname{Bernoulli}(1/2)\), and \(\varepsilon\sim N(0,1)\) independently, and set
\begin{align*}
\tau(X)&=0.20+0.65\psi_1(X)-0.55\psi_2(X)+0.10(X_3-0.5),\\
Y&=1+X_4+0.5X_5^2-0.5X_6+W\tau(X)+\varepsilon.
\end{align*}
Here \(\psi_j(x)=2\one\{x_j\ge1/2\}-1\) for \(j\in\{1,2\}\). The first three covariates determine treatment effects, and the remaining three enter the baseline outcome. The two step components have opposite signs, so learning both is necessary to recover the optimal treatment rule.

All inputs use dense ranks on the pooled training and evaluation array. Representation \(R_h\) adds \(c-1\) copies of the rank column for \(X_h\), where \(h\in\{1,2\}\) and \(c\in\{5,10,40,400\}\). Each input has \(M=c+5\) columns. The other step-component covariate remains a singleton, and the input without repeated columns has \(c=1\). Repetition changes the chances that each covariate enters the split search while preserving the six split classes and the routing of evaluation observations.

I cross fixed, default, and one-third \texttt{mtry} settings,
\[
\texttt{mtry}=2,\qquad
\texttt{mtry}=\min\{\lceil\sqrt M\rceil+20,M\},\qquad
\texttt{mtry}=\lfloor M/3\rfloor,
\]
with minimum node sizes \(5,10,20\). Each \texttt{grf} 2.4.0 forest uses 2,000 trees, sample fraction \(1/2\), and honesty fraction \(1/2\). For \(c=5,10,40,400\), default \texttt{mtry} equals \(10,15,27,41\), and one-third \texttt{mtry} equals \(3,5,15,135\), respectively. For class sampling, the same \texttt{mtry} rules are applied to the six split classes.

Each setting is evaluated on ten independent training samples of 3,000 observations and a common independent evaluation array of 10,000 observations. Within each replication, the fits share the training data and forest seed. They also use the true outcome mean \(\E(Y\mid X)\) and propensity \(\Prb(W=1\mid X)=1/2\). Using these known functions isolates the representation comparisons from nuisance-estimation error.

MSE measures treatment-effect estimation error, regret measures the value lost through incorrect treatment decisions, and \(\mathrm{SR}\) measures sign reversals between the two representations. On the common evaluation array \(X_1^e,\ldots,X_N^e\), with \(N=10{,}000\), define
\[
\begin{aligned}
\mathrm{Reg}=\frac1N\sum_{i=1}^N|\tau(X_i^e)|
 \one\{\widehat\pi(X_i^e)\ne\pi^*(X_i^e)\},\qquad
\mathrm{SR}=\frac{100}{N}\sum_{i=1}^N
 \one\{\widehat\tau_1(X_i^e)\widehat\tau_2(X_i^e)<0\},
\end{aligned}
\]
Here \(\widehat\tau=\widehat\tau_h\) for representation \(R_h\), \(\widehat\pi(x)=\one\{\widehat\tau(x)>0\}\), and \(\pi^*(x)=\one\{\tau(x)>0\}\). Table~\ref{tab:mc} reports MSE and regret for each representation and \(\mathrm{SR}\) between them.

Small \texttt{mtry} produces substantial policy disagreement even at moderate multiplicity. With minimum node size five and \(c=10\), \(\mathrm{SR}\) is 32.9 percent under fixed \texttt{mtry} and 1.2 percent under the default. At \(c=400\), raising \texttt{mtry} from its default of 41 to 135 lowers \(\mathrm{SR}\) from 40.0 to 16.9 percent. For \(c=5,10,40\), the one-third rule gives smaller numerical values of \texttt{mtry} than the default and greater disagreement. These comparisons link sensitivity to how often the singleton can enter the split search.

Minimum node size affects the number of opportunities to select the singleton along a tree path. With default \texttt{mtry} at \(c=400\), increasing minimum node size from five to twenty shortens mean realized paths from about 6.8 to 4.0 splits. The fraction of sampled paths using the singleton falls from 45.8 to 30.0 percent in \(R_1\) and from 48.6 to 33.8 percent in \(R_2\). \(\mathrm{SR}\) rises from 40.0 to 47.7 percent, and the two MSEs increase from 0.1533 and 0.1699 to 0.1882 and 0.2227. For the input without repeated columns and default \texttt{mtry}, the same node-size change lowers MSE from 0.0312 to 0.0282. Larger leaves can improve estimation while leaving fewer searches in which a rarely offered covariate can be selected. These path summaries use 32 prespecified trees and 128 prespecified evaluation points per fit after honesty pruning.

Class sampling maps every representation to the same six-column input. With common randomization, it therefore reproduces the \(c=1\) MSE and regret at each \texttt{mtry} and node-size setting in Table~\ref{tab:mc}. Its \(\mathrm{SR}\) is zero throughout the multiplicity grid.

\begin{table}[p]
\centering\setstretch{1}
\sbox{\mcgridpage}{\begin{minipage}{\linewidth}
\caption{Representation multiplicity, \texttt{mtry}, and minimum node size} \label{tab:mc}
\centering\fontsize{10}{12}\selectfont
\setlength{\tabcolsep}{3pt}
\renewcommand{\arraystretch}{1.10}
\sbox{\mcgridbody}{\begin{tabular}{@{}r*{5}{r}@{\hspace{10pt}}*{5}{r}@{\hspace{10pt}}*{5}{r}@{}}
\toprule
 & \multicolumn{15}{c}{\texttt{mtry}}\\[-2pt]
 & \multicolumn{5}{c}{2} & \multicolumn{5}{c}{Default} & \multicolumn{5}{c}{\(\lfloor M/3\rfloor\)}\\
\cmidrule(lr){2-6}\cmidrule(lr){7-11}\cmidrule(lr){12-16}
 & \(\mathbf{SR}\) & \multicolumn{2}{c}{MSE \(\times100\)} & \multicolumn{2}{c}{Regret \(\times100\)} & \(\mathbf{SR}\) & \multicolumn{2}{c}{MSE \(\times100\)} & \multicolumn{2}{c}{Regret \(\times100\)} & \(\mathbf{SR}\) & \multicolumn{2}{c}{MSE \(\times100\)} & \multicolumn{2}{c}{Regret \(\times100\)}\\
\(c\) & (\%) & \(R_1\) & \(R_2\) & \(R_1\) & \(R_2\) & (\%) & \(R_1\) & \(R_2\) & \(R_1\) & \(R_2\) & (\%) & \(R_1\) & \(R_2\) & \(R_1\) & \(R_2\)\\
\midrule
\multicolumn{16}{@{}l}{\textit{Panel A. Minimum node size 5}}\\[4pt]
1 & 0.0 & 5.59 & 5.59 & 0.98 & 0.98 & 0.0 & 3.12 & 3.12 & 1.04 & 1.04 & 0.0 & 5.59 & 5.59 & 0.98 & 0.98\\[-3pt]
 & {\scriptsize (0.0)} & {\scriptsize (0.31)} & {\scriptsize (0.31)} & {\scriptsize (0.16)} & {\scriptsize (0.16)} & {\scriptsize (0.0)} & {\scriptsize (0.21)} & {\scriptsize (0.21)} & {\scriptsize (0.20)} & {\scriptsize (0.20)} & {\scriptsize (0.0)} & {\scriptsize (0.31)} & {\scriptsize (0.31)} & {\scriptsize (0.16)} & {\scriptsize (0.16)}\\[3pt]
5 & 23.8 & 7.46 & 7.28 & 1.92 & 1.98 & 1.3 & 3.07 & 2.97 & 1.05 & 1.04 & 17.4 & 5.56 & 4.92 & 1.62 & 1.41\\[-3pt]
 & {\scriptsize (2.7)} & {\scriptsize (0.41)} & {\scriptsize (0.57)} & {\scriptsize (0.15)} & {\scriptsize (0.69)} & {\scriptsize (0.2)} & {\scriptsize (0.21)} & {\scriptsize (0.21)} & {\scriptsize (0.19)} & {\scriptsize (0.20)} & {\scriptsize (2.4)} & {\scriptsize (0.38)} & {\scriptsize (0.44)} & {\scriptsize (0.19)} & {\scriptsize (0.49)}\\[3pt]
10 & \textbf{32.9} & 11.14 & 11.78 & 2.23 & 3.51 & \textbf{1.2} & 3.06 & 2.97 & 1.04 & 1.04 & 16.6 & 5.24 & 4.47 & 1.58 & 1.38\\[-3pt]
 & {\scriptsize (2.6)} & {\scriptsize (0.47)} & {\scriptsize (0.67)} & {\scriptsize (0.09)} & {\scriptsize (0.78)} & {\scriptsize (0.2)} & {\scriptsize (0.21)} & {\scriptsize (0.20)} & {\scriptsize (0.19)} & {\scriptsize (0.20)} & {\scriptsize (2.4)} & {\scriptsize (0.38)} & {\scriptsize (0.45)} & {\scriptsize (0.19)} & {\scriptsize (0.49)}\\[3pt]
40 & 46.9 & 22.20 & 27.13 & 2.54 & 6.97 & 5.2 & 3.47 & 3.09 & 1.12 & 1.04 & 16.7 & 5.25 & 4.47 & 1.59 & 1.35\\[-3pt]
 & {\scriptsize (0.6)} & {\scriptsize (0.32)} & {\scriptsize (0.66)} & {\scriptsize (0.05)} & {\scriptsize (0.15)} & {\scriptsize (0.9)} & {\scriptsize (0.24)} & {\scriptsize (0.24)} & {\scriptsize (0.19)} & {\scriptsize (0.23)} & {\scriptsize (2.3)} & {\scriptsize (0.37)} & {\scriptsize (0.45)} & {\scriptsize (0.18)} & {\scriptsize (0.49)}\\[3pt]
400 & 49.0 & 32.93 & 43.35 & 2.84 & 7.85 & \textbf{40.0} & 15.33 & 16.99 & 2.39 & 5.15 & \textbf{16.9} & 5.29 & 4.47 & 1.60 & 1.39\\[-3pt]
 & {\scriptsize (0.1)} & {\scriptsize (0.21)} & {\scriptsize (0.28)} & {\scriptsize (0.13)} & {\scriptsize (0.13)} & {\scriptsize (1.9)} & {\scriptsize (0.44)} & {\scriptsize (0.73)} & {\scriptsize (0.06)} & {\scriptsize (0.56)} & {\scriptsize (2.4)} & {\scriptsize (0.38)} & {\scriptsize (0.46)} & {\scriptsize (0.18)} & {\scriptsize (0.50)}\\
\midrule
\multicolumn{16}{@{}l}{\textit{Panel B. Minimum node size 10}}\\[4pt]
1 & 0.0 & 6.22 & 6.22 & 0.79 & 0.79 & 0.0 & 2.90 & 2.90 & 0.91 & 0.91 & 0.0 & 6.22 & 6.22 & 0.79 & 0.79\\[-3pt]
 & {\scriptsize (0.0)} & {\scriptsize (0.33)} & {\scriptsize (0.33)} & {\scriptsize (0.13)} & {\scriptsize (0.13)} & {\scriptsize (0.0)} & {\scriptsize (0.20)} & {\scriptsize (0.20)} & {\scriptsize (0.17)} & {\scriptsize (0.17)} & {\scriptsize (0.0)} & {\scriptsize (0.33)} & {\scriptsize (0.33)} & {\scriptsize (0.13)} & {\scriptsize (0.13)}\\[3pt]
5 & 27.2 & 8.35 & 8.30 & 2.13 & 2.20 & 1.1 & 2.84 & 2.74 & 0.91 & 0.90 & 20.1 & 5.90 & 5.27 & 1.78 & 1.42\\[-3pt]
 & {\scriptsize (3.0)} & {\scriptsize (0.43)} & {\scriptsize (0.57)} & {\scriptsize (0.14)} & {\scriptsize (0.82)} & {\scriptsize (0.2)} & {\scriptsize (0.20)} & {\scriptsize (0.19)} & {\scriptsize (0.18)} & {\scriptsize (0.17)} & {\scriptsize (2.8)} & {\scriptsize (0.38)} & {\scriptsize (0.46)} & {\scriptsize (0.21)} & {\scriptsize (0.56)}\\[3pt]
10 & 36.2 & 12.31 & 13.61 & 2.39 & 4.01 & 1.0 & 2.83 & 2.74 & 0.90 & 0.88 & 19.3 & 5.54 & 4.72 & 1.74 & 1.37\\[-3pt]
 & {\scriptsize (2.8)} & {\scriptsize (0.46)} & {\scriptsize (0.65)} & {\scriptsize (0.05)} & {\scriptsize (0.86)} & {\scriptsize (0.2)} & {\scriptsize (0.20)} & {\scriptsize (0.20)} & {\scriptsize (0.18)} & {\scriptsize (0.17)} & {\scriptsize (2.6)} & {\scriptsize (0.38)} & {\scriptsize (0.48)} & {\scriptsize (0.20)} & {\scriptsize (0.55)}\\[3pt]
40 & 48.7 & 22.96 & 28.71 & 2.58 & 7.45 & 6.4 & 3.32 & 2.89 & 1.09 & 0.89 & 19.4 & 5.49 & 4.76 & 1.73 & 1.42\\[-3pt]
 & {\scriptsize (0.2)} & {\scriptsize (0.25)} & {\scriptsize (0.54)} & {\scriptsize (0.05)} & {\scriptsize (0.06)} & {\scriptsize (1.3)} & {\scriptsize (0.25)} & {\scriptsize (0.25)} & {\scriptsize (0.21)} & {\scriptsize (0.21)} & {\scriptsize (2.7)} & {\scriptsize (0.40)} & {\scriptsize (0.48)} & {\scriptsize (0.21)} & {\scriptsize (0.56)}\\[3pt]
400 & 49.6 & 31.54 & 42.22 & 2.64 & 7.75 & 44.1 & 16.47 & 18.83 & 2.51 & 6.06 & 19.9 & 5.55 & 4.77 & 1.78 & 1.37\\[-3pt]
 & {\scriptsize (0.0)} & {\scriptsize (0.13)} & {\scriptsize (0.20)} & {\scriptsize (0.06)} & {\scriptsize (0.05)} & {\scriptsize (1.4)} & {\scriptsize (0.39)} & {\scriptsize (0.68)} & {\scriptsize (0.04)} & {\scriptsize (0.43)} & {\scriptsize (2.7)} & {\scriptsize (0.41)} & {\scriptsize (0.50)} & {\scriptsize (0.20)} & {\scriptsize (0.57)}\\
\midrule
\multicolumn{16}{@{}l}{\textit{Panel C. Minimum node size 20}}\\[4pt]
1 & 0.0 & 8.07 & 8.07 & 0.57 & 0.57 & 0.0 & 2.82 & 2.82 & 0.81 & 0.81 & 0.0 & 8.07 & 8.07 & 0.57 & 0.57\\[-3pt]
 & {\scriptsize (0.0)} & {\scriptsize (0.35)} & {\scriptsize (0.35)} & {\scriptsize (0.13)} & {\scriptsize (0.13)} & {\scriptsize (0.0)} & {\scriptsize (0.20)} & {\scriptsize (0.20)} & {\scriptsize (0.17)} & {\scriptsize (0.17)} & {\scriptsize (0.0)} & {\scriptsize (0.35)} & {\scriptsize (0.35)} & {\scriptsize (0.13)} & {\scriptsize (0.13)}\\[3pt]
5 & 31.0 & 10.35 & 10.98 & 2.30 & 2.77 & 1.3 & 2.71 & 2.60 & 0.82 & 0.78 & 25.0 & 7.30 & 6.77 & 2.09 & 1.77\\[-3pt]
 & {\scriptsize (3.3)} & {\scriptsize (0.44)} & {\scriptsize (0.62)} & {\scriptsize (0.10)} & {\scriptsize (1.00)} & {\scriptsize (0.4)} & {\scriptsize (0.19)} & {\scriptsize (0.20)} & {\scriptsize (0.19)} & {\scriptsize (0.16)} & {\scriptsize (3.0)} & {\scriptsize (0.45)} & {\scriptsize (0.55)} & {\scriptsize (0.18)} & {\scriptsize (0.73)}\\[3pt]
10 & 42.2 & 14.82 & 17.18 & 2.48 & 5.53 & 1.0 & 2.71 & 2.61 & 0.81 & 0.77 & 24.3 & 6.70 & 5.91 & 2.06 & 1.64\\[-3pt]
 & {\scriptsize (2.0)} & {\scriptsize (0.41)} & {\scriptsize (0.65)} & {\scriptsize (0.03)} & {\scriptsize (0.64)} & {\scriptsize (0.3)} & {\scriptsize (0.19)} & {\scriptsize (0.19)} & {\scriptsize (0.18)} & {\scriptsize (0.16)} & {\scriptsize (2.9)} & {\scriptsize (0.46)} & {\scriptsize (0.56)} & {\scriptsize (0.19)} & {\scriptsize (0.67)}\\[3pt]
40 & 49.1 & 24.43 & 31.43 & 2.59 & 7.55 & 8.7 & 3.42 & 2.92 & 1.19 & 0.73 & 24.2 & 6.66 & 5.98 & 2.05 & 1.64\\[-3pt]
 & {\scriptsize (0.2)} & {\scriptsize (0.19)} & {\scriptsize (0.45)} & {\scriptsize (0.05)} & {\scriptsize (0.04)} & {\scriptsize (1.9)} & {\scriptsize (0.26)} & {\scriptsize (0.27)} & {\scriptsize (0.23)} & {\scriptsize (0.17)} & {\scriptsize (3.0)} & {\scriptsize (0.44)} & {\scriptsize (0.57)} & {\scriptsize (0.19)} & {\scriptsize (0.70)}\\[3pt]
400 & 49.6 & 31.10 & 42.06 & 2.63 & 7.77 & \textbf{47.7} & 18.82 & 22.27 & 2.54 & 7.13 & 24.7 & 6.78 & 5.94 & 2.07 & 1.72\\[-3pt]
 & {\scriptsize (0.1)} & {\scriptsize (0.12)} & {\scriptsize (0.21)} & {\scriptsize (0.05)} & {\scriptsize (0.06)} & {\scriptsize (0.5)} & {\scriptsize (0.40)} & {\scriptsize (0.59)} & {\scriptsize (0.04)} & {\scriptsize (0.16)} & {\scriptsize (3.0)} & {\scriptsize (0.46)} & {\scriptsize (0.57)} & {\scriptsize (0.19)} & {\scriptsize (0.70)}\\
\bottomrule
\end{tabular}}
\ifdim\wd\mcgridbody>\linewidth
  \PackageError{mcgrid}{Simulation table width (\the\wd\mcgridbody) exceeds the page width (\the\linewidth)}{Reduce table spacing.}
\fi
\usebox{\mcgridbody}
\legend{Means over ten paired Monte Carlo replications, with Monte Carlo standard errors in parentheses. In \(R_h\), \(X_h\) has multiplicity \(c\), giving \(M=c+5\) columns. The \(c=1\) rows have no repeated columns. \(\mathrm{SR}\) is the percentage of sign reversals between \(R_1\) and \(R_2\). MSE and regret are multiplied by 100. Bold entries highlight comparisons discussed in the text.}
\end{minipage}}
\ifdim\dimexpr\ht\mcgridpage+\dp\mcgridpage\relax>\textheight
  \PackageError{mcgrid}{Simulation table height (\the\dimexpr\ht\mcgridpage+\dp\mcgridpage\relax) exceeds the page height (\the\textheight)}{Reduce table spacing.}
\fi
\usebox{\mcgridpage}
\end{table}

\FloatBarrier

\subsection{Job Training and Earnings}

The job-training application examines representation sensitivity in the random-forest outcome regression used to estimate the average treatment effect on the treated (ATET). I replicate specifications 1--3 of \citet{chernozhukov2022automatic} for the National Supported Work (NSW) experimental comparison and the Panel Study of Income Dynamics (PSID) and Current Population Survey (CPS) comparison groups. Using the original input matrices and replication environment, the all-column fits reproduce eight of the nine published ATETs and all nine standard errors at two decimal places. The CPS specification 3 ATET differs by one cent. The retained samples contain 357 observations in NSW, 912 in PSID, and 6,089 in CPS.

In specifications 1 and 2, age, education, 1974 earnings, and 1975 earnings enter the outcome regression alongside their squares. Their raw observed supports are nonnegative, so each raw-square pair generates the same CART threshold partitions. I form classes on the pooled observed and treatment-counterfactual array, which includes both treatment states for every covariate vector used in prediction. In specifications 1 and 2, retaining one original column per class removes four or five of the 24 to 30 input columns while preserving these partitions. Specification 3 contains four degree-one orthogonal polynomial terms that are affine transformations of the original covariates. Floating-point errors break ties in these terms. I remove them and their treatment interactions by construction, reducing 344 columns to 336 while leaving all retained values unchanged.

The comparisons refit the random-forest outcome regression and hold the samples, folds, Riesz/Lasso estimates, and orthogonal-score formula fixed. Within each fold, the paired forests start from the same saved random state in the original replication. Here \(\widehat m(w,x)\) estimates \(\E(Y\mid W=w,X=x)\), and the individual outcome-regression contrast is \(\widehat m(1,x)-\widehat m(0,x)\).

Table~\ref{tab:empirical} compares the original forest, two reduced-column fits using original values, and class sampling. The default reduction recomputes \texttt{mtry} as one third of the supplied column count, rounded down. The fixed comparison retains its all-column numerical value, holding tuning fixed as columns are removed. In specification 3, class sampling assigns the known equivalent pairs identical ranks before forming classes. It uses the lexicographically ordered canonical rank matrix in Section~\ref{sec:quotient}, with \texttt{mtry} equal to one third of the number of classes, rounded down. All fits use Intel R 4.0.3 and the historical \texttt{randomForest} 4.6-14 binary, with 1,000 trees, bootstrap sampling, minimum node size five, and five folds. The paired fits use seed 1. Four additional forest seeds with all columns retained provide benchmarks for changes in ATET estimates and individual contrast signs.

Specifications 1 and 2 have 24 and 30 columns. Reduction and class sampling retain 19 and 25 columns in NSW and CPS, and 20 and 26 in PSID. Their default \texttt{mtry} values are 6 and 8. The fixed reduction uses 8 and 10. Specification 3 has 344 columns in the original fit and 336 after removal. Its reduced default and class-sampled fits use \texttt{mtry} equal to 112, and its fixed comparison retains 114.

\begin{table}[p]
\centering\setstretch{1}
\sbox{\empiricalpage}{\begin{minipage}{\linewidth}
\caption{ATET estimates and representation sensitivity} \label{tab:empirical}
\centering\fontsize{10.5}{12}\selectfont
\setlength{\tabcolsep}{4pt}
\renewcommand{\arraystretch}{1.00}
\sbox{\empiricalbody}{\begin{tabular*}{\linewidth}{@{\extracolsep{\fill}}lcrrrrrrrrr@{}}
\toprule
 & & \shortstack{All\\columns} & \multicolumn{4}{c}{Reduced columns} & \multicolumn{2}{c}{\shortstack{Class\\sampling}} & \multicolumn{2}{c}{\shortstack{Seed\\benchmark}}\\
\cmidrule(lr){4-7}
 & & & \multicolumn{2}{c}{Default \texttt{mtry}} & \multicolumn{2}{c}{Fixed \texttt{mtry}} & & & & \\
\cmidrule(lr){3-3}\cmidrule(lr){4-5}\cmidrule(lr){6-7}\cmidrule(lr){8-9}\cmidrule(lr){10-11}
Sample & Spec. & ATET & ATET & SR & ATET & SR & ATET & SR & \(|\Delta\mathrm{ATET}|\) & SR\\
 & & (\$) & (\$) & (\%) & (\$) & (\%) & (\$) & (\%) & (\$) & (\%)\\
\midrule
NSW & 1 & 3,107 & 3,119 & 5.9 & 3,084 & 5.6 & 3,098 & 6.7 & 34 & 2.4\\[-1pt]
 &  & {\footnotesize (1,327)} & {\footnotesize (1,327)} &  & {\footnotesize (1,314)} &  & {\footnotesize (1,319)} &  & {\footnotesize (49)} & {\footnotesize (2.5)}\\[-1pt]
 &  &  & {\footnotesize [+12]} &  & {\footnotesize [-23]} &  & {\footnotesize [-9]} &  &  & \\[4pt]
 & 2 & 3,077 & 3,024 & 4.5 & 3,102 & 2.5 & 3,223 & 6.4 & 88 & 2.2\\[-1pt]
 &  & {\footnotesize (1,319)} & {\footnotesize (1,266)} &  & {\footnotesize (1,287)} &  & {\footnotesize (1,389)} &  & {\footnotesize (150)} & {\footnotesize (3.4)}\\[-1pt]
 &  &  & {\footnotesize [-53]} &  & {\footnotesize [+25]} &  & {\footnotesize [+145]} &  &  & \\[4pt]
 & 3 & 2,785 & 2,803 & 2.0 & 2,773 & 1.4 & 2,760 & 2.0 & 55 & 1.7\\[-1pt]
 &  & {\footnotesize (819)} & {\footnotesize (820)} &  & {\footnotesize (828)} &  & {\footnotesize (820)} &  & {\footnotesize (88)} & {\footnotesize (2.5)}\\[-1pt]
 &  &  & {\footnotesize [+18]} &  & {\footnotesize [-12]} &  & {\footnotesize [-26]} &  &  & \\[2pt]
\midrule
PSID & 1 & 1,522 & 1,525 & 4.3 & 1,654 & 3.5 & 1,567 & 4.4 & 24 & 1.7\\[-1pt]
 &  & {\footnotesize (977)} & {\footnotesize (923)} &  & {\footnotesize (949)} &  & {\footnotesize (925)} &  & {\footnotesize (65)} & {\footnotesize (2.0)}\\[-1pt]
 &  &  & {\footnotesize [+3]} &  & {\footnotesize [+132]} &  & {\footnotesize [+45]} &  &  & \\[4pt]
 & 2 & 1,337 & 1,367 & 4.4 & 1,370 & 3.1 & 1,440 & 4.1 & 9 & 1.9\\[-1pt]
 &  & {\footnotesize (956)} & {\footnotesize (922)} &  & {\footnotesize (930)} &  & {\footnotesize (917)} &  & {\footnotesize (12)} & {\footnotesize (2.1)}\\[-1pt]
 &  &  & {\footnotesize [+31]} &  & {\footnotesize [+33]} &  & {\footnotesize [+103]} &  &  & \\[4pt]
 & 3 & 2,011 & 2,014 & 5.7 & 2,064 & 5.5 & 2,111 & 4.6 & 52 & 4.4\\[-1pt]
 &  & {\footnotesize (988)} & {\footnotesize (997)} &  & {\footnotesize (995)} &  & {\footnotesize (1,003)} &  & {\footnotesize (108)} & {\footnotesize (4.5)}\\[-1pt]
 &  &  & {\footnotesize [+4]} &  & {\footnotesize [+54]} &  & {\footnotesize [+101]} &  &  & \\[2pt]
\midrule
CPS & 1 & 1,640 & 1,497 & 10.1 & 1,612 & 5.8 & 1,495 & 9.6 & 14 & 4.4\\[-1pt]
 &  & {\footnotesize (616)} & {\footnotesize (602)} &  & {\footnotesize (609)} &  & {\footnotesize (600)} &  & {\footnotesize (23)} & {\footnotesize (5.0)}\\[-1pt]
 &  &  & {\footnotesize [-143]} &  & {\footnotesize [-28]} &  & {\footnotesize [-145]} &  &  & \\[4pt]
 & 2 & 1,584 & 1,489 & 11.6 & 1,510 & 9.1 & 1,480 & 11.1 & 11 & 6.2\\[-1pt]
 &  & {\footnotesize (616)} & {\footnotesize (607)} &  & {\footnotesize (609)} &  & {\footnotesize (607)} &  & {\footnotesize (23)} & {\footnotesize (6.7)}\\[-1pt]
 &  &  & {\footnotesize [-95]} &  & {\footnotesize [-74]} &  & {\footnotesize [-104]} &  &  & \\[4pt]
 & 3 & 1,907 & 1,924 & 10.4 & 1,918 & 11.7 & 1,933 & 11.9 & 9 & 9.4\\[-1pt]
 &  & {\footnotesize (652)} & {\footnotesize (651)} &  & {\footnotesize (651)} &  & {\footnotesize (650)} &  & {\footnotesize (13)} & {\footnotesize (12.0)}\\[-1pt]
 &  &  & {\footnotesize [+17]} &  & {\footnotesize [+11]} &  & {\footnotesize [+26]} &  &  & \\[2pt]
\bottomrule
\end{tabular*}}
\ifdim\wd\empiricalbody>\linewidth
  \PackageError{empiricaltables}{Table 2 exceeds the page width}{Reduce table spacing.}
\fi
\usebox{\empiricalbody}
\legend{Dollar amounts are rounded to the nearest dollar. Parentheses below ATET estimates give standard errors. Brackets give changes from the all-column fit, computed before rounding. SR is the percentage of retained observations whose individual contrast \(\widehat m(1,x)-\widehat m(0,x)\) changes sign relative to that fit. Reduced columns use either recomputed default \texttt{mtry} or the all-column numerical value. Specification 3 removes eight known equivalent columns, giving 336. Its class sampling assigns the known pairs identical ranks. Seed benchmarks hold folds fixed and compare seed 1 with forest seeds 2--5, reporting mean absolute ATET changes and mean SR, with maxima in parentheses.}
\end{minipage}}
\ifdim\dimexpr\ht\empiricalpage+\dp\empiricalpage\relax>\textheight
  \PackageError{empiricaltables}{Table 2 exceeds the page height}{Reduce table spacing.}
\fi
\usebox{\empiricalpage}
\end{table}

In specifications 1 and 2, reducing columns reverses individual contrast signs for 2.5--11.6 percent of observations. Eleven of the twelve comparisons produce more sign changes than the largest seed benchmark in Table~\ref{tab:empirical}. The exception is NSW specification 2 with fixed \texttt{mtry}. Absolute ATET changes range from 3 to 143 dollars and are all below 0.24 baseline standard errors. Each ATET remains positive and retains its original significance classification at the 5 percent level. Individual contrasts are sensitive to column allocation even when aggregate conclusions remain stable.

In specification 3, removing the eight known equivalent columns reverses individual contrast signs for 1.4--11.7 percent of observations. Both reduced-column fits exceed the largest seed benchmark in PSID, while the NSW and CPS changes remain within their seed benchmarks. Absolute ATET changes range from 4 to 54 dollars. All estimates remain positive and significant at the 5 percent level.

With the known equivalent pairs assigned identical ranks in specification 3, class sampling produces identical outcome-regression predictions, ATET estimates, and reported standard errors from the full and reduced inputs in all nine cases. The canonical matrices also coincide after column permutation, sign reversal, and duplicate addition. Relative to the original all-column forest, class sampling changes individual contrast signs for 2.0--11.9 percent of observations. The zero sign-change rate between class-sampled representations illustrates Proposition~\ref{prop:class}, while the table measures each method's departure from the original forest.

The \href{https://github.com/YiNiu-bot/representation-multiplicity-causal-forests}{replication repository} provides replication code and implementation details.

\FloatBarrier

\section{Conclusion}\label{sec:conclusion}

Feature subsampling makes column allocation part of a forest's specification. Under the theorem's conditions, representation multiplicity can change limiting CATEs and treatment rules without changing covariate information. Recovery depends on split access accumulated over a tree path. Sampling split-equivalent classes removes representation dependence on the declared array under canonical routing and common randomness. The Monte Carlo shows how \texttt{mtry} and minimum node size jointly shape representation sensitivity. With default \texttt{mtry}, larger leaves improve the fit without repeated columns while shortening paths and amplifying sensitivity under limited candidate access. Class sampling preserves the original six-column forest's performance across equivalent representations.

\clearpage
\begin{appendix}

\section{Proofs}\label{app:proofs}

\begin{proof}[Proof of Theorem~\ref{thm:main}]
All probability bounds for the endpoints are uniform over \(a\in\mathcal C_a\). Write \(\mathcal D_n\) for the training data and the fitted nuisance functions, including their training randomization. Let \(\Theta_n\) denote the randomization of a generic ordinary tree at sample size \(n\). When \(\omega_n<1\), let \(\Theta_n^0\) denote the corresponding locally conditioned failure-tree randomization.

Let \(\mathcal R\) be the class of axis-aligned rectangles in \([0,1]^2\) with positive \(P_X\)-mass, and let \(I\) range over the root splitting and estimation index sets of the two forests. For \(C\in\mathcal R\), put
\[
N_I(C)=\sum_{i\in I}\one\{X_i\in C\},\qquad
\widehat p_I(C)=\frac{N_I(C)}{|I|},\qquad
\mu(C)=\E(\Gamma_i\mid X_i\in C).
\]
Root allocations are independent of the observations. Each root sample therefore has the i.i.d. law of \((X_i,\Gamma_i)\). The rectangle indicators and their products with the bounded oracle score are Vapnik--Chervonenkis (VC) type classes. Relative concentration bounds \citep[Sections 2 and 3]{gine2006ratio}, followed by a union bound over the root samples, give, for \(\ell_n=\log(B_ns_n)\) and a fixed constant \(K_\Gamma\), with probability tending to one,
\[
\begin{aligned}
&\left|\widehat p_I(C)-P_X(C)\right|
+\left|\frac1{|I|}\sum_{i\in I}\Gamma_i\one\{X_i\in C\}
-P_X(C)\mu(C)\right|\\
&\hspace{2em}\le K_\Gamma\left\{
\sqrt{\frac{P_X(C)\ell_n}{|I|}}+\frac{\ell_n}{|I|}\right\}.
\end{aligned}
\]
The bound holds for every root sample \(I\) and \(C\in\mathcal R\). Since \(|I|\ge\eta s_n\) and \(\ell_n/k_n\to0\), inversion of this bound yields, for every fixed \(\zeta>0\) and \(b_n=\sqrt{\ell_n/k_n}+\ell_n/k_n\),
\begin{equation}\label{eq:proof-relative}
\max_I\sup_{\substack{C\in\mathcal R\\N_I(C)\ge\zeta k_n}}
\left\{
\left|\frac{\widehat p_I(C)}{P_X(C)}-1\right|
+\left|\frac{\sum_{i\in I}\Gamma_i\one\{X_i\in C\}}{N_I(C)}
-\mu(C)\right|\right\}=O_p(b_n).
\end{equation}
These bounds hold over all such rectangles before node transitions are chosen. They therefore apply to ordinary and failure trees with the same root samples. Strictly monotone transformations preserve rectangular cells in the original coordinates.

The score identity and fitted-score perturbation satisfy
\[
\mu(C)=\E\{\tau_a(X)\mid X\in C\},\qquad
\max_i|\widehat\Gamma_{i,n}-\Gamma_i|\le2\epsilon_n.
\]
For a split of \(C\) into \(C_L\) and \(C_R\), define
\[
Q_C(j,t)
=P_X(C_L\mid C)P_X(C_R\mid C)
\{\mu(C_L)-\mu(C_R)\}^2.
\]
Every admissible child contains at least \(k_n\) splitting observations. The relative cell-mass bound controls both child shares, and the score-mean bound controls their contrast. Boundedness then gives
\begin{equation}\label{eq:proof-criterion}
\sup_{C,j,t}
\left|\widehat Q_{C,n}(j,t)-Q_C(j,t)\right|
=O_p(b_n+\epsilon_n)=o_p(1).
\end{equation}
The supremum is over reached nodes and admissible splits in the original coordinates.

First consider a node before coordinate \(j\) has been used. Its range in that coordinate is still \([0,1]\), and independence removes the other component from the split contrast. Direct calculation gives
\begin{equation}\label{eq:proof-shape}
Q_C(j,t)=
\begin{cases}
a_j^2t/(1-t),&0<t\le1/2,\\
a_j^2(1-t)/t,&1/2\le t<1.
\end{cases}
\end{equation}
The unique maximizer is \(1/2\), with gain \(a_j^2\). Positivity of \(\underline a\) and the coefficient gap imply \(a_2^2-a_1^2\ge2\underline a\Delta_a>0\). At nodes with at least \(C_0k_n\) splitting observations, split thresholds approaching \(1/2\) are size-admissible and balanced on the concentration event. Maximization and \eqref{eq:proof-criterion} thus imply
\begin{equation}\label{eq:proof-localization}
\max_{j,F}|\widehat t_{j,F,n}-1/2|=o_p(1),
\end{equation}
Here \(\widehat t_{j,F,n}\) is the selected threshold at a node \(F\) where coordinate \(j\) is used for the first time. The maximum ranges over such nodes above this size band in the ordinary and failure trees.

Continuation and balance supply the depth bound. If \(n_{d,n}(x)\) is the splitting-sample size at depth \(d\) on the path to \(x\), then, whenever a split is made,
\[
\beta n_{d,n}(x)\le n_{d+1,n}(x)\le(1-\beta)n_{d,n}(x).
\]
The root size is proportional to \(s_n\), no path stops above \(C_0k_n\), and every child has at least \(k_n\) observations. Consequently there are constants \(0<c_\beta<C_\beta<\infty\) and an integer \(K\), independent of \(n,a,h\), such that, with probability tending to one,
\begin{equation}\label{eq:proof-depth}
c_\beta H_n\le J_{h,n}(X),J_{h,n}^{\mathrm{fail}}(X)\le C_\beta H_n.
\end{equation}
Here \(J_{h,n}\) is the number of candidate draws on an ordinary path, and \(J_{h,n}^{\mathrm{fail}}\) is defined when \(\omega_n<1\). At most \(K\) draws occur after the path enters the band \(n_{d,n}(x)<C_0k_n\), including a final unsuccessful search. Reached leaves have splitting sizes between \(k_n\) and \(C_0k_n\). Let \(I_{h,b,n}^{\rm split}\), \(I_{h,b,n}^{\rm est}\), and \(\mathcal L_{h,b,n}\) denote the splitting indices, estimation indices, and terminal cells of tree \(b\le B_n\) under representation \(h\), and put \(N_{h,b,n}^{\rm est}(C)=\sum_{i\in I_{h,b,n}^{\rm est}}\one\{X_i\in C\}\). Applying \eqref{eq:proof-relative} with \(\zeta=1\) to the splitting samples gives, uniformly over terminal cells,
\[
P_X(C)=\frac{N_{I_{h,b,n}^{\rm split}}(C)}{|I_{h,b,n}^{\rm split}|}
\{1+O_p(b_n)\}.
\]
Since \(|I_{h,b,n}^{\rm est}|\ge\eta s_n\) and \(|I_{h,b,n}^{\rm split}|\le(1-\eta)s_n\), \(|I_{h,b,n}^{\rm est}|P_X(C)\) is bounded below by \(\{1+O_p(b_n)\}\eta k_n/(1-\eta)\). The unrestricted rectangle-mass bound above then yields, uniformly over these cells,
\[
\frac{N_{h,b,n}^{\rm est}(C)}{|I_{h,b,n}^{\rm est}|P_X(C)}
=1+O_p(b_n).
\]
Thus, for \(c_\eta=\eta/\{2(1-\eta)\}>0\),
\[
\Prb\left\{
\min_{\substack{h=1,2,\ b\le B_n\\C\in\mathcal L_{h,b,n}}}
N_{h,b,n}^{\rm est}(C)\ge c_\eta k_n
\right\}\longrightarrow1.
\]
Applying \eqref{eq:proof-relative} with \(\zeta=c_\eta\) and the fitted-score perturbation bound gives
\begin{equation}\label{eq:proof-leaf}
\max_{\substack{h=1,2,\ b\le B_n\\C\in\mathcal L_{h,b,n}}}\left|
\frac{1}{N_{h,b,n}^{\rm est}(C)}
\sum_{\substack{i\in I_{h,b,n}^{\rm est}\\X_i\in C}}\widehat\Gamma_{i,n}
-\E\{\tau_a(X)\mid X\in C\}
\right|=O_p(b_n+\epsilon_n)=o_p(1).
\end{equation}
The counts are positive on this event.

We next control the approximation after first use. Choose a deterministic \(\delta_n\downarrow0\) so that \eqref{eq:proof-localization} is at most \(\delta_n\) with probability tending to one. For each coordinate, its first-use nodes above the terminal band form an antichain. Label only their children by the sign suggested by the side of the split threshold, and retain this label in their descendants. The total population mass receiving a wrong label is
\[
\sum_F P_X(F)|\widehat t_{j,F,n}-1/2|\le\delta_n,
\]
where the sum is over these above-band first-use nodes. For a descendant \(C\), let \(q(C)\) be the conditional fraction with the wrong label. Then
\[
\sup_t Q_C(j,t)\le4\overline a^2q(C),
\qquad
\int_C|\E\{\psi_j(X)\mid X\in C\}-\psi_j(x)|\,dP_X(x)\le4P_X(C)q(C).
\]
These bounds also apply after further splits. For \(u_n=\sqrt{\delta_n}\), select the first descendant on each path for which \(q(C)>u_n\). Those descendants are disjoint, so their combined mass is at most \(\delta_n/u_n\). Outside this vanishing set, a previously used coordinate has gain at most \(4\overline a^2u_n\) at every subsequent node. An offered unused coordinate has gain at least \(\underline a^2-o_p(1)\) above the terminal band and therefore wins. Only the root can involve competition between two unused coordinates.

The multiply represented coordinate is offered with probability at least \(1/2\) at each draw. Together with \eqref{eq:proof-depth}, the preceding ordering argument bounds its probability of remaining unused above the terminal band by
\[
2^{-\lfloor c_\beta H_n\rfloor+K+1}+o(1)=o(1).
\]
For the singleton, the probability of first use only in that band is also \(o(1)\). If \(\omega_n\le H_n^{-1/2}\), it is at most \(K\omega_n\). Otherwise, avoiding an earlier use has probability at most \(\exp\{-\omega_n(\lfloor c_\beta H_n\rfloor-K-1)\}+o(1)\). The error terms here include the exceptional target mass just bounded.

Let \(D_{h,n}(x)\) indicate first use of the singleton in a generic tree, and let \(\widehat T^{\rm tree}_{h,n}\) be that tree's prediction. If \(D_{h,n}(x)=0\), its leaf has no restriction in the singleton coordinate, whose conditional mean is exactly zero. For a first use above the terminal band, localization and the wrong-label bound control the integrated approximation error. A late first use adds at most twice its probability to the integrated error for the sign component, which vanishes by the preceding bound. The same bound applies to the component associated with the multiply represented coordinate. Thus
\begin{equation}\label{eq:proof-mixture}
\E_{\Theta_n}\!\left[
\left\|\widehat T^{\rm tree}_{h,n}-a_h\psi_h
-a_{3-h}D_{h,n}\psi_{3-h}\right\|_{L^1(P_X)}
\middle|\mathcal D_n\right]=o_p(1).
\end{equation}
Suppose now that \(\omega_n<1\). Conditional on \(Q_{M_n}=q\), the singleton is included with probability \(q/M_n\), so its unconditional inclusion probability at a node is \(\omega_n\). On the locally conditioned failure tree, write \(\rho_{h,n,t}^0(x)\) for the probability that the ordinary transition at the current node chooses the singleton as its splitting coordinate. It satisfies \(0\le\rho_{h,n,t}^0(x)\le\omega_n\). After the root and above the terminal band, this probability equals \(\omega_n\) outside the exceptional paths. With \(\Lambda_{h,n}^0(x)=\sum_{t=1}^{J_{h,n}^{\mathrm{fail}}(x)}\rho_{h,n,t}^0(x)\), we obtain
\begin{equation}\label{eq:proof-hazard}
\left|\Lambda_{h,n}^0(x)-J_{h,n}^{\mathrm{fail}}(x)\omega_n\right|
\le(K+1)\omega_n
\end{equation}
outside a set of joint training, failure-tree, and test-point probability tending to zero.

Node locality ensures that conditioning transitions on other branches leaves the transition kernels on the target path unchanged. At each stage, factor the probability of avoiding a singleton split from the transition kernel conditional on that event. Multiplication along the path and integration over the conditional kernels give the exact identity
\begin{equation}\label{eq:proof-survival}
u_{h,n}(x)=\Prb_{\Theta_n}\{D_{h,n}(x)=0\mid\mathcal D_n\}
=\E_{\Theta_n^0}\!\left[
\prod_{t=1}^{J_{h,n}^{\mathrm{fail}}(x)}(1-\rho_{h,n,t}^0(x))
\middle|\mathcal D_n\right].
\end{equation}
The kernels are well-defined because their conditioning probabilities are at least \(1-\omega_n>0\).

If \(\mathcal I_n\to0\), then \(\omega_n\to0\) and \(\Lambda_{h,n}^0(X)\to_p0\) by \eqref{eq:proof-depth} and \eqref{eq:proof-hazard}. The product converges to one since its deficit is at most \(\Lambda_{h,n}^0(X)\). If \(\mathcal I_n\to\infty\) and \(\omega_n<1\), the cumulative hazard diverges and
\[
\prod_{t=1}^{J_{h,n}^{\mathrm{fail}}(X)}(1-\rho_{h,n,t}^0(X))
\le e^{-\Lambda_{h,n}^0(X)}\longrightarrow_p0.
\]
For indices with \(\omega_n=1\), both coordinates are always offered. The stronger one wins at the root and the other at the next interior node, outside the same exceptional set. Thus recovery also holds on these indices, without a failure-tree conditioning on a null event.

For the intermediate case, fix \(a\in\mathcal C_a\). The condition \(\mathcal I_n\to\kappa\in(0,\infty)\) implies \(\omega_n\to0\). Equations~\eqref{eq:intermediate-depth} and \eqref{eq:proof-hazard} yield \(\Lambda_{h,n}^0(X)\to_p\kappa d_h(a)\). Moreover,
\[
0\le-\log\prod_{t=1}^{J_{h,n}^{\mathrm{fail}}(X)}(1-\rho_{h,n,t}^0(X))
-\Lambda_{h,n}^0(X)
\le\frac{\omega_n\Lambda_{h,n}^0(X)}{1-\omega_n}
\longrightarrow_p0.
\]
The product therefore converges to \(e^{-\kappa d_h(a)}\). All products are bounded. Joint convergence with an independent test point implies convergence of \(u_{h,n}\) in \(L^1(P_X)\), in probability over \(\mathcal D_n\), to one, zero, or \(e^{-\kappa d_h(a)}\) in the respective regimes.

Finally, conditional on \(\mathcal D_n\), the independent partition of the \(B_n\) trees into groups of bounded size, together with bounded score predictions, gives
\[
\left\|\widehat T_{h,n}-
\E_{\Theta_n}(\widehat T^{\rm tree}_{h,n}\mid\mathcal D_n)\right\|_{L^1(P_X)}
=O_p(B_n^{-1/2})=o_p(1).
\]
Combining this bound with \eqref{eq:proof-mixture} and the three limits for \(u_{h,n}\) proves \eqref{eq:selection-limit}, \eqref{eq:recovery-limit}, and \eqref{eq:intermediate-targets}.

Strictly monotone transformations preserve unordered threshold partitions after reversing the orientation of decreasing columns. The bounds above are expressed in the original coordinates and are independent of the number, orientation, and form of these transformations.
\end{proof}

\begin{proof}[Proof of Corollary~\ref{cor:sign}]
In the selection regime, Markov's inequality and \eqref{eq:selection-limit} give, for \(h=1,2\),
\[
\Prb_X\{\operatorname{sgn}(\widehat T_{h,n}(X))\ne\psi_h(X)\}
\le\frac{2}{\underline a}
\|\widehat T_{h,n}-a_h\psi_h\|_{L^1(P_X)}=o_p(1).
\]
The independent signs \(\psi_1(X)\) and \(\psi_2(X)\) disagree with probability one half. On that event, \(|\tau_a(X)|=a_2-a_1\ge\Delta_a\). A second application of Markov's inequality shows that both fitted magnitudes exceed \(\underline a/2\) outside a set of \(P_X\)-mass \(o_p(1)\). This proves \eqref{eq:sign-reversal} and the stated separation from zero.

For any rule \(\pi\),
\[
\operatorname{Reg}_a(\pi)
=\E\!\left[|\tau_a(X)|
\one\{\pi(X)\ne\pi_a^*(X)\}\right].
\]
Since \(a_2>a_1\), \(\pi_a^*(x)=\one\{\psi_2(x)=1\}\), so the limit of \(\widehat\pi_{2,n}\) has zero regret. For the other representation,
\[
\operatorname{Reg}_a(\one\{\psi_1(\cdot)=1\})
=\E\!\left[\tau_a(X)
\{\one\{\psi_2(X)=1\}-\one\{\psi_1(X)=1\}\}\right]
=\frac{a_2-a_1}{2}.
\]
Boundedness and classification convergence transfer these values to the estimated rules, proving \eqref{eq:regret}.
\end{proof}

\begin{proof}[Proof of Corollary~\ref{cor:boundary}]
The capped-Poisson rule satisfies
\[
\E Q_M=\xi_M+e^{-\xi_M}-\E(N_M-M)_+.
\]
For fixed \(\xi_M=\xi\), the final term vanishes as \(M\to\infty\), giving \(\omega_n\sim(\xi+e^{-\xi})/M_n\). If \(\xi_M\to\infty\) and \(\xi_M/M\to0\), Poisson concentration gives \(\E Q_M/\xi_M\to1\). Under the package default, \(\xi_M/\sqrt M\to1\), so \(\omega_n\sim M_n^{-1/2}\). The two boundaries follow by substituting these expressions into \(\mathcal I_n=H_n\omega_n\). For fixed finite \(M\), \(\omega_n>0\) is constant and \(H_n\to\infty\).
\end{proof}

\begin{proof}[Proof of Proposition~\ref{prop:class}]
Fix \(\theta\). By construction,
\[
\mathcal C(R)=\mathcal C(R')
\quad\Longrightarrow\quad R^{\rm can}=R'^{\rm can}.
\]
Consequently, tuning parameters and fixed or canonically fitted nuisance functions coincide. Common randomization gives identical splitting and estimation index sets for each tree. Let \(A_{b,v}(R,\theta)\) contain the row indices routed to node \(v\) of tree \(b\). The root contains \(\{1,\ldots,N\}\) in both representations. Suppose the two index sets coincide at a reached node. The canonical values, candidate classes, admissible splits, scores, and stopping decisions then coincide. If the node splits, canonical orientation and common tie-breaking imply
\[
\bigl(A_{b,v0}(R,\theta),A_{b,v1}(R,\theta)\bigr)
=\bigl(A_{b,v0}(R',\theta),A_{b,v1}(R',\theta)\bigr).
\]
Induction proves equality of the ordered node partitions.

Index the training observations by the first \(n\) rows. Let \(I_b^{\rm est}(\theta)\) be tree \(b\)'s estimation indices and \(L_b(\ell,R,\theta)\) its leaf containing row \(\ell\). The forest weights are
\[
\alpha_{\ell i}(R,\theta)=\frac1B\sum_{b=1}^B
\frac{\one\{i\in I_b^{\rm est}(\theta)\cap L_b(\ell,R,\theta)\}}
{|I_b^{\rm est}(\theta)\cap L_b(\ell,R,\theta)|},
\qquad i\le n,\quad\ell\le N,
\]
with an empty-leaf contribution defined as zero. Equality of the leaves and estimation indices gives \(\boldsymbol\alpha(R,\theta)=\boldsymbol\alpha(R',\theta)\). Predictions and variance calculations use identical inputs and randomization. Policy ranking also uses the same observation order for ties. Hence \(\mathcal F(R,\theta)=\mathcal F(R',\theta)\).
\end{proof}

\end{appendix}

\end{document}